\documentclass[11pt,letterpaper]{article}
\usepackage{amsmath, amsthm, amssymb}
\usepackage{graphicx}
\usepackage[latin1]{inputenc}
\usepackage{fullpage}
\usepackage[pdftitle={Tile},
            pdfauthor={Author},	
            pdfkeywords={Keywords}]{hyperref}
\usepackage{color}
\usepackage{bm}
\usepackage{algorithmicx}
\usepackage[noend]{algpseudocode}
\usepackage{algorithm}

\usepackage{pdfpages}
\usepackage{xspace}

\newtheorem{lemma}{Lemma}

\newtheorem{theorem}{Theorem}

\newcommand{\R}{\mathbb{R}}

\newcommand{\remove}[1]{}

\newcommand{\uvec}{\mathbf{u}}
\newcommand{\vvec}{\mathbf{v}}

\newcommand{\cvec}{\mathbf{c}}

\title{Minimization of Gini impurity via connections with the $k$-means problem}
\author{Eduardo Laber \\PUC-Rio, Brazil \\ {\tt laber@inf.puc-rio.br}  \and
Lucas Murtinho \\  PUC-Rio, Brazil \\ {\tt lucas.murtinho@gmail.com }}

\begin{document}
	\maketitle

\thispagestyle{empty}

\begin{abstract}
The Gini impurity is one of the measures used to
select attribute in Decision Trees/Random Forest construction.
In this note we discuss connections between the problem of computing
the partition with minimum Weighted Gini impurity and the $k$-means clustering problem.
Based on these connections we show  that the
computation of the partition with minimum Weighted Gini is a  NP-Complete problem and 
we also discuss how to obtain  new algorithms with provable approximation  for the Gini Minimization problem. 
\end{abstract}

\pagebreak
\setcounter{page}{1}
\section{Introduction}
Decision Trees and Random Forests  are among the most popular 
methods for classification tasks.
It is widely known that decision trees, specially small ones, are easy to interpret
while random forest usually yield
 to more stable/accurate classifications.


A key decision during the construction of these structures
is the selection of the attribute that is used for branching at each node.
The standard approach for  this selection   is to
evaluate  the ability of each attribute to generate 'pure' partitions, that is, partitions in which each branch is very homogeneous with respect to the class distribution
of its examples. To measure how impure each branch is, impurity measures are often employed.
An impurity measure   maps a vector $\uvec=(u_1,\ldots,u_k)$, counting how many examples of each class we have in a node (branch), into a non-negative scalar \footnote{In the original definition an impurity measure  maps a vector
of probabilities into a non-negative scalar.}.
Arguably, two of the most classical impurity measures are the \emph{Gini} impurity  $$i_{Gini}(\uvec) = \sum_{i=1}^k \frac{u_i}{\|\uvec\|_1}\left(1-\frac{u_i}{\|\uvec\|_1}\right),$$ which is
 used in the CART package \cite{Breiman84}, and the \emph{Entropy} impurity $$i_{Entr}(\uvec) = - \sum_{i=1}^k \frac{u_i}{\| \uvec \|_1} \log \left( \frac{u_i}{  \|\uvec\|_1 } \right ),$$ that along with its variants is used in the C4.5 decision tree inducer
 \cite{quinlan:c4-5:92}.

Given an attribute, the goal is then to find a split
for the attribute values  that induces a  partition of the set of examples with minimum weighted  impurity, where the weights are given by the number of examples that lie into
each of the  branches.
	
Here,  we discuss connections between the problem of 
computing the partition with minimum weighted Gini and the $k$-means clustering problem.


\section{Connections between Gini minimization and  $k$-means clustering}
For a vector $\vvec$  where all components are non-negative the weighted Gini impurity
$Gini(\vvec)$ is defined as  $Gini(\vvec) = \|\vvec\|_1 \cdot i_{Gini}(\vvec)$.
Let  $A$ be a nominal attribute  that may take  $n$ possible values $a_1,\ldots,a_n$.
The  $k$-ary  Partition with Minimum
Weighted Gini Problem ({\em $k$-PMWGP})  can be described abstractly as follows.
We are given a collection of $n$ vectors $V \subset \R^d$,
where the $i$th component of the $j$th vector
counts the number of examples 
in class $i$ for which the attribute $A$ has value $a_j$.
The goal is to find a partition  ${\cal P}$ of $V$ into $k$ disjoint groups of 	 vectors $V_1,\ldots,V_k$   so as to minimize
the  sum of the weighted Gini impurities 
\begin{equation}
\label{eq:kmeans-goal1}
Gini({\cal P})= \sum_{i=1}^{k} Gini\bigg(\sum_{ \vvec \in V_i} \vvec \bigg).
\end{equation}

Recently, we obtained simple constant approximation algorithms for this problem:
an $O(n \log n + nd )$ time $2$-approximation for the case
where $k=2$ \cite{conf/icml/LMP18} and a linear time $3$-approximation for arbitrary $k$
\cite{2018arXiv180705241C}. In fact these papers also handle a more general class of impurity measures that includes the Entropy impurity.
The complexity of ({\em $k$-PMWGP}) remained open.

A problem that is equivalent to the above problem from the perspective of optimality but it is different from the perspective of approximation is the problem of finding the partition 
${\cal P}$ of $V$ into $k$ groups that minimizes 
\begin{equation}
\label{eq:kmeans-goal2}
Gini({\cal P}) - \sum_{\vvec \in V} Gini(\vvec).
\end{equation} 
Using concavity properties  of Gini one can prove that the above expression is always non-negative.  An $\alpha$-approximation with respect to goal (\ref{eq:kmeans-goal2}) implies an
$\alpha$-approximation with respect to goal  (\ref{eq:kmeans-goal1}) but the converse is not necessarily true, so that approximations with respect to goal (\ref{eq:kmeans-goal2}) are stronger.

In  the geometric $k$-means problem we are given a set of vectors $V$ in $\R^d$ and
the goal is to find a partition 
${\cal P}$ of
$V$ into $k$ groups $V_1,\ldots,V_k$ and
 set of $k$ centers $\cvec_1,\ldots,\cvec_k$ in $\R^d$   
such that 
$$ Cost_{KM}({\cal P})= \sum_{i=1}^k \sum_{\vvec  \in V_i} \| \vvec- \cvec_i \|_2^2$$ is
minimized.  

It is well known that if $U$ is a set of vectors then the
vector $\cvec$ for which $\sum_{\vvec \in U} \| (\vvec - \cvec)\|_2^2$ is minimum is
the centroid  of $U$, that is, 
$\cvec = ( \sum_{ \vvec \in U} \vvec ) / |U|$.

\medskip

We argue that the  following  connections between {\em $k$-PMWGP} and $k$-means hold:

\begin{itemize}
\item[C1] Let $V$ be an instance of $k$-means where all vectors  have the same $\ell_1$ norm.
If ${\cal P}$ is an optimal partition for instance $V$ then
${\cal P}$ is also an optimal partition for instance $V$ of   {\em $k$-PMWGP}

\item[C2] there exists a pseudo-polynomial time 
reduction from  {\em $k$-PMWGP} to the geometric $k$-means problem   

\end{itemize}

The key observation for establishing C1 and C2 is the following lemma.

\begin{lemma}
Let $X$ be a set of vectors, all of them with  $\ell_1$ norm equal to $L$.
Then, 
$$ Gini\bigg(\sum_{ \vvec \in X} \vvec \bigg) - \sum_{\vvec \in X} Gini(\vvec)=
L \times \left (  \sum_{\vvec  \in X}   \| \vvec- \cvec \|_2^2 \right ) , 
$$
where $\cvec$ is the centroid of the set of vectors in $X$.
\label{lem:equivalence}
\end{lemma}
\begin{proof}
Let $\uvec=\sum_{\vvec \in X} \vvec$.
We have that
$$ Gini(\uvec) -   \sum_{\vvec \in X} Gini(\vvec) =
\| \uvec \|_1 \left ( \sum_{i=1}^d \left (  1 - \frac{u_i}{\|\uvec\|_1}  \right) \left ( \frac{u_i}{\|\uvec\|_1}  \right) \right ) -
\sum_{ \vvec \in X }  \sum_{i=1}^d  \| \vvec \|_1   
\left (  1 - \frac{v_i}{\|\vvec\|_1}  \right) \left ( \frac{v_i}{\|\vvec\|_1}  \right)  $$

On the other hand,
$$ \sum_{\vvec  \in X}   \| \vvec- \cvec \|_2^2 = \sum_{i=1}^d  \sum_{\vvec \in X}  (v_i - c_i)^2 $$
Thus, it suffices  to show that 
for any $i$ 

$$
\| \uvec \|_1 
\left (  1 - \frac{u_i}{\|\uvec\|_1}  \right) \left ( \frac{u_i}{\|\uvec\|_1}  \right)  -
\sum_{ \vvec \in X } \|\vvec \|_1 \left (  1 - \frac{v_i}{\|\vvec\|_1}  \right) \left ( \frac{v_i}{\|\vvec\|_1}  \right)
= L \left ( \sum_{\vvec \in X}  (v_i - c_i)^2  \right )$$

The left side is 
equal to 
 $$ u_i  - \frac{(u_i)^2}{ \|\uvec \|_1} - 
 \left (
 \sum_{\vvec \in X} v_i  - |X| \frac{\sum_{\vvec \in X}(v_i)^2 }{ \|\uvec\|_1} \right )
 = \frac{ |X| \sum_{\vvec \in X}(v_i)^2 }{\|\uvec\|_1} - \frac{(u_i)^2}{ \| \uvec \|_1}=
\frac{  \sum_{\vvec \in X}(v_i)^2 }{L} - \frac{(u_i)^2}{ L|X|}
  $$

Moreover,   the
 righthand side  is equal to 

$$   \sum_{\vvec \in X} (v_i)^2 -  \frac{ ( \sum_{\vvec \in X} v_i ) ^2}{|X|}
= \sum_{\vvec \in X}(v_i)^2 - \frac{(u_i)^2}{ |X|}, $$
which established the lemma.
\end{proof}

We should note that the result presented in the previous lemma  is mentioned in the appendix
of \cite{Chou91} where the Gini index is discussed.

The connection C1 is a direct consequence of Lemma \ref{lem:equivalence}
since it implies that for all $k$-partitions
${\cal P}$ of $V$ 
$$Gini({\cal P})=L \cdot Cost_{KM}({\cal P})+\sum_{\vvec \in V} Gini(\vvec) $$

From the connection   C1 and the hardness of geometric $k$-means established 
in \cite{awasthi2015hardness} we obtain:

\begin{theorem}
The Partition with Minimum Weighted Gini Problem ({\em PMWGP}) is  NP-Complete with
respect to goal (\ref{eq:kmeans-goal1}) and APX-Hard with respect to goal (\ref{eq:kmeans-goal2}). 
\end{theorem}
\begin{proof}
The result follows from \cite{awasthi2015hardness}, where a polynomial time reduction
from the vertex cover problem on triangle free graphs  to the $k$-means problem is presented.
In this reduction, given a graph $G=(V,E)$, every edge $e$ in $E$ is mapped into
a vector $\vvec$ in $\R^{|V|}$ where the $i$-th component $v_i$ is $1$ if $i$ is incident on $e$ and it is 0,
otherwise.   It is proved  that if  the minimum vertex cover of $G$ has size $k$
then the optimum cost of the  corresponding $k$-means problem is at most $|E|-k$ and if the minimum vertex cover has size at least $(1+\epsilon)k$ then the minimum cost is
at least $|E|- (1-\Omega(\epsilon)k$. 

Our result follows from Lemma \ref{lem:equivalence} and from the fact that in the instance of $k$-means above described all  vectors have $\ell_1$ norm equals 2.
\end{proof}

With regards to the connection C2, let $V$ be an  input of {\em $k$-PMWGP} and 
let  $V'$ be an instance of  $k$-means obtained from $V$ as follows:
for each vector  $\vvec \in V$ we add  to the input set $V'$ exactly $\| \vvec \|_1$  copies of vector $\vvec'=\vvec / \| \vvec \|_1$.
Using Lemma \ref{lem:equivalence} and also the fact that  in any optimal solution for $k$-means identical  vectors
are in the same partition,  we conclude 
that the optimum value of $V$ and $V'$ differ by  exactly $ \sum_{\vvec \in V} Gini(\vvec)$.
Note that  instance $V'$ is obtained from $V$ in pseudo-poytime.

\remove{

From this reduction one can   obtain pseudo-polynomial time algorithms for
{\em $k$-PMWGP} with provable approximation guarantee with
respect to the objective function (\ref{eq:kmeans-goal2}). 
As an example, if we apply the PTAS for $k$-means proposed in  \cite{Kumar:2004:SLT} to the  reduced instance $V'$ we obtain a randomized
 approximation scheme for {\em $k$-PMWGP} that runs is pseudo-polytime when $k$ is fixed.
On the other hand, by applying the algorithm 
 proposed in  \cite{KanungoEtAl04} to the  reduced instance $V'$ we obtain a
$(9+\epsilon)$  approximation  for {\em $k$-PMWGP} that runs is pseudo-polytime
for arbitrary $k$.
}


From this reduction one we can  obtain
new algorithms for  {\em $k$-PMWGP} with provable approximation.
As an example, we discuss how to  obtain a PTAS for {\em $k$-PMWGP}
with respect to the objective function (\ref{eq:kmeans-goal2}) when $k$ is fixed.
First, in our reduction, we keep  each distinct vector and its  multiplicity 
rather than all the $\sum_{\vvec \in V} \| \vvec \|_1$ vectors of instance $V'$. Thus, we can construct the instance
$V'$ from $V$ in polytime.
Next, we run over instance $V'$ an adapted version of  the
recursive PTAS for $k$-means proposed in  \cite{Kumar:2004:SLT}.
This  version   efficiently handles copies of the same vector and it is 
explained
refering to the presentation of the PTAS that is given in Figure 1 of  \cite{journals/talg/AckermannBS10}.

Let $W=\sum_{\vvec \in V} \| \vvec \|_1$.  
The adapted version is as follows:

\begin{enumerate}
\item The set of vectors in $V'$ is represented using the distinct vectors and its multiplicities.
\item In the step 6 of the  algorithm described in  Figure 1  of \cite{journals/talg/AckermannBS10}  a constant number of vectors is sampled from
a set of at most $n$ vectors. In our adaptation we sample from 
 a set of $W$ vectors, with at most $n$ of them being distinct.
 This incurs an extra $O( \log W)$ factor to the running time.
\item At  Step 12 of the same  Figure
one needs to compute, from a set of vectors $R$,  the $|R|/2$ closest vectors 
 to a given set of centroids $C$. Since this computation can be
performed with time complexity proportional to the number of distinct vectors in $R$, rather than in $O(|R|)$ time, we do not incur any additional cost.
\end{enumerate}

This adapted version also incurs an extra factor of  $O(\log (W) ^k)$  with respect to the original one (executed over $n$  vectors) due to the number of nodes in the recursion tree. Thus, it  runs in polynomial time when $k$ is fixed. 
Without efficiently handling the copies we would have a pseudo-polynomial time
approximation scheme.

\remove{

}



\end{document}